\documentclass[11pt,a4paper,reqno]{amsart}%
\usepackage{amsthm,amsmath,amsfonts,amssymb,amsxtra,appendix,bookmark,dsfont,latexsym}

\makeatletter
\usepackage{hyperref}
\usepackage{delarray,a4,color}
\usepackage{euscript}
\usepackage[latin1]{inputenc}

\theoremstyle{plain}
\newtheorem{theorem}{Theorem}[section]
\newtheorem{lemma}{Lemma}[section]

\theoremstyle{remark}
\newtheorem{remark}{Remark}[section]

\numberwithin{equation}{section}

\DeclareMathOperator{\Tr}{Tr}
\DeclareMathOperator{\tr}{Tr}

\def\geqslant{\ge}
\def\leqslant{\le}
\def\bq{\begin{eqnarray}}
\def\eq{\end{eqnarray}}
\def\bqq{\begin{eqnarray*}}
\def\eqq{\end{eqnarray*}}

\def\wto{\rightharpoonup}

\newcommand{\norm}[1]{\left\lVert #1 \right\rVert}

\newcommand\1{{\ensuremath {\mathds 1} }}

\renewcommand{\epsilon}{\varepsilon}

\def\NN {\mathbb{N}}

\def\cS {\mathcal{S}}

\def\gS{\mathfrak{S}}
\def\S {\mathcal{S}}

\def\cS {\mathcal{S}}

\def\gH{\mathfrak{H}}

\newcommand\pscal[1]{{\ensuremath{\left\langle #1 \right\rangle}}}

\renewcommand{\leq}{\leqslant}
\renewcommand{\geq}{\geqslant}

\newcommand{\Gammat}{\tilde{\Gamma}}

\newcommand{\gHNs}{\gH ^{N} }

\newcommand{\gHks}{\gH ^{k} }

\newcommand{\muu}{\mu ^{\rm up}}
\newcommand{\mul}{\mu ^{\rm low}}

\title[Quantum de Finetti theorem for bosons]{Remarks on the quantum de Finetti theorem for bosonic systems}

\author[M. Lewin]{Mathieu LEWIN}
\address{CNRS \& Laboratoire de Math\'ematiques (UMR 8088), Universit\'e de Cergy-Pontoise, F-95000 Cergy-Pontoise, France.}
\email{mathieu.lewin@math.cnrs.fr}

\author[P.T. Nam]{Phan Th\`anh NAM}
\address{CNRS \& Laboratoire de Math\'ematiques (UMR 8088), Universit\'e de Cergy-Pontoise, F-95000 Cergy-Pontoise, France.}
\email{phan-thanh.nam@u-cergy.fr}

\author[N. Rougerie]{Nicolas ROUGERIE}
\address{Universit\'e Grenoble 1 \& CNRS,  LPMMC (UMR 5493), B.P. 166, F-38 042 Grenoble, France}
\email{nicolas.rougerie@grenoble.cnrs.fr}

\date{August 21, 2014}

\begin{document}

\begin{abstract}
The quantum de Finetti theorem asserts that the $k$-body density matrices of a $N$-body bosonic state approach a convex combination of Hartree states (pure tensor powers) when $N$ is large and $k$ fixed. In this note we review a construction due to Christandl, Mitchison, K\"onig and Renner \cite{ChrKonMitRen-07} valid for finite dimensional Hilbert spaces, which gives a quantitative version of the theorem. We first propose a variant of their proof that leads to a slightly improved estimate. Next we provide an alternative proof of an explicit formula due to Chiribella \cite{Chiribella-11}, which gives the density matrices of the constructed state as a function of those of the original state.       
\end{abstract}

\maketitle

\tableofcontents

\section{Introduction}\label{sec:intro}
 
Consider a system of $N$ bosons with one-particle state space $\gH$, a separable Hilbert space. Pure states of this system are rank-one projectors $|\Psi\rangle \langle \Psi|$ onto normalized vectors $\Psi \in \gHNs$, the symmetric tensor product of $N$ copies of $\gH$. Mixed states are then convex combinations of pure states, that is, self-adjoint positive trace-class operators with trace~$1$:
\begin{equation}\label{eq:states}
\cS (\gHNs):= \left\{ \Gamma \in \gS ^1 (\gHNs),\: \Gamma = \Gamma ^*, \: \Gamma \geq 0, \: \tr \Gamma = 1 \right\}. 
\end{equation}
It is a well-known fact that Hartree states, i.e. projections onto tensor powers $\Psi=u ^{\otimes N}\in \gH^N$ (for $u$ a normalized vector in $\gH$), play a very special role in the physics of bosonic systems. Indeed, since bosons, contrarily to fermions, do not satisfy the Pauli exclusion principle, there is a possibility for many particles to occupy the same quantum state, which is the meaning of the ansatz~$u ^{\otimes N}$. It is in fact the case for non interacting particles in the ground state of a one-body Hamiltonian, or in the thermal state at low enough temperature. This is the famous \emph{Bose-Einstein condensation} phenomenon. 

The case of interacting bosonic particles is much more subtle. It is no longer true that the bosonic ground state of an interacting Hamiltonian (say with two-body interactions) is given exactly by a Hartree state. In fact, one may expect that this is almost the case for large systems with properly scaled interactions, that is in the limit $N\to \infty$ with a $N$-dependent two-body coupling ensuring that the total energy of the system stays of order $N$. Proofs of this expectation for various models are available in the literature (see \cite{PetRagVer-89,RagWer-89,Werner-92,LieSeiSolYng-05,LewNamRou-13} and references therein), emphasizing the relevance of mean-field theories obtained from the ansatz $\Psi = u ^{\otimes N}$ to the study of large bosonic systems.  

One possible strategy to investigate these mean-field limits was employed in \cite{PetRagVer-89,RagWer-89,Werner-92,LewNamRou-13} and is based on a remarkable structure property of the set of bosonic states in the limit of infinitely many particles. Roughly speaking, \emph{any $N$-body bosonic state is almost a convex combination of Hartree states for large $N$}. Since the quantum mechanical energy is linear in the state\footnote{That is, in the density matrix $|\Psi\rangle\langle \Psi|$. It is of course quadratic in the wave-function $\Psi$.}, it is then immediate to guess that the infimum over all states should coincide with the infimum over Hartree states in the large $N$ limit. Hence the validity of the mean-field approximation follows using very little of the particular form of the Hamiltonian. Of course, to obtain a rigorous proof, one must justify the preceding heuristic statement in a sufficiently strong sense, and this is precisely the purpose of the \emph{quantum de Finetti theorem}. 

\medskip

Let us start with a sequence of (mixed) states $\Gamma_N \in \cS (\gHNs)$ and construct the corresponding reduced $k$-body density matrices
\begin{equation}\label{eq:reduced matrices}
\Gamma_N ^{(k)} := \tr_{k+1\to N} \left[ \Gamma_N \right] \in \cS (\gHks). 
\end{equation}
Here $\tr_{k+1\to N}$ is a notation for the partial trace with respect to all the variables but $k$. Modulo a diagonal extraction one can always assume that there is a subsequence (not relabeled) along which
\begin{equation}\label{eq:weak conv}
\Gamma_N ^{(k)} \wto_* \Gamma ^{(k)}
\end{equation}
weakly-$\ast$ in the trace-class, when $N\to \infty$.
% \[
%  \tr_{\gHks} \left[ \Gamma_N ^{(k)} K_k \right] \to  \tr_{\gHks} \left[ \Gamma ^{(k)} K_k \right]
% \]
% for any compact operator $K_k$ on $\gHks$. 
If in addition the convergence in \eqref{eq:weak conv} is strong (that is, holds in trace-class norm), then we have the consistency relations 
\begin{equation}\label{eq:consistent}
\Gamma ^{(k)} = \tr_{k+1} \Gamma ^{(k+1)}
\end{equation}
for any $k=0,1,2,\ldots $. The hierarchy $(\Gamma ^{(k)})_{k\in \NN}$ can then be thought of as describing a system with infinitely many particles and the quantum de Finetti theorem of St\o{}rmer~\cite{Stormer-69} and Hudson-Moody \cite{HudMoo-75} states that there exists a unique Borel probability measure $\mu$ on the sphere $S\gH$ of the one-particle Hilbert space, invariant under the group action of $S^1$, such that 
\begin{equation}\label{eq:deF strong}
\boxed{\Gamma ^{(k)} = \int_{S\gH} |u ^{\otimes k}\rangle \langle u ^{\otimes k} | d\mu (u).}  
\end{equation}
The above result is a generalization of the famous classical de Finetti theorem of Hewitt and Savage \cite{DeFinetti-31,DeFinetti-37,HewSav-55}.
% , which has been used to study the mean-field limit of classical mechanics systems~\cite{BraHep-77,Spohn-81,MesSpo-82,CagLioMarPul-92,Kiessling-93}. 
It is one way to give a rigorous meaning to the previous heuristic statement, already extremely useful in applications to mean-field systems: typically it is safe to assume that particles interact only pairwise and then the energy of a state depends only on the reduced $2$-body density matrix. 

A second step is to wonder whether one could obtain a quantitative version of the de Finetti theorem, much like Diaconis and Freedman obtained a quantitative version of the classical de Finetti theorem~\cite{DiaFre-80}. Namely, one may ask the following question: given a state $\Gamma_N \in \cS (\gHNs)$, does there exist another state  $\Gammat_N$ of the form
\[
 \Gammat_N = \int_{S\gH} |u ^{\otimes N}\rangle \langle u ^{\otimes N} | d\mu_N (u)
\]
with $\mu_N$ a probability measure, such that 
\begin{equation}\label{eq:quant deF}
 \boxed{ \tr_{\gHks} \left| \Gamma_N ^{(k)} - \Gammat_N ^{(k)} \right| \leq C(N,k)}
\end{equation}
with $C(N,k)\to 0$ when $N\to \infty$ and $k$ is fixed ? The answer to this question is still open for general Hilbert spaces, but important progress has been made in recent years in the case of finite dimensional spaces \cite{KonRen-05,FanVan-06,ChrKonMitRen-07,Chiribella-11,Harrow-13}. Motivated by applications in  quantum information theory, variants of the above question have also been investigated \cite{ChrKonMitRen-07,ChrTon-09,RenCir-09,Renner-07,BraHar-12} by either adding assumptions on the state or weakening the conclusion.

As concerns \eqref{eq:quant deF}, the best estimate available, proved by Christandl, König, Mitchison and Renner in \cite{ChrKonMitRen-07}, is $C(N,k)=4dk/N$ where\footnote{Note the convention in \cite{ChrKonMitRen-07} where the trace norm is divided by $2$.} $d$ is the dimension of $\gH$. It has the desired property that for fixed $d$ and $k$, $C(N,k) \to 0$ when $N\to \infty$. We emphasize that here we discuss only bosonic states, but in \cite{ChrKonMitRen-07} it is shown that one may generalize \eqref{eq:quant deF} to general symmetric states (boltzons), with $C(N,k) = 4kd ^2 /N$. A result of this form had appeared before in \cite{FanVan-06} for the case $k=2$.  

\medskip

In this note we do three things: first we review the construction of \cite{ChrKonMitRen-07} and discuss some related heuristics in Section \ref{sec:motiv}. We then propose a variant of their proof which gives a slightly better bound $\sim 2 kd /N$ when $N\gg dk$, see Theorem \ref{thm:CKMR}. Next we express the construction, which is in fact given by an anti-Wick quantization, in terms of creation/annihilation operators. This yields our Theorem \ref{thm:CKMR-identity}, an explicit expression of the density matrices of the state $\Gammat_N$ as a function of the density matrices of the original state $\Gamma_N$. A variant of the bound follows immediately, see Remark \ref{rm:non opt bound}. 

The explicit formula we obtain for the density matrices of the state $\Gammat_N$ was derived before by Chiribella in \cite{Chiribella-11} (see also \cite{Harrow-13}), where it is interpreted as a relation between ``universal measure and prepare quantum channels'' and ``optimal cloning maps''. Our proof is new, and proceeds by expressing the problem at hand in terms of the second quantization formalism and the CCR algebra. This method, along with the relationship between de Finetti measures and upper/lower symbols in a coherent state representation that we discuss in Section \ref{sec:motiv}, provides an alternative interpretation of the fundamental principles at work.

\medskip

We end this introduction by noting that, for finite dimensional spaces, the quantum de Finetti theorem in the form \eqref{eq:deF strong} can be obtained from the aforementioned bounds simply by passing to the limit $N\to \infty$ at fixed $d$. Then, as we proved in \cite{LewNamRou-13} (see Section 2 therein), a stronger theorem for general spaces can be obtained by combining the localization method of \cite{Lewin-11} and \eqref{eq:deF strong} in finite dimensional spaces. Namely, under the sole assumption \eqref{eq:weak conv}, the conclusion in \eqref{eq:deF strong} still holds provided the measure $\mu$ is allowed to live on the ball of the Hilbert space rather than on the sphere. This is connected to recent results of Ammari and Nier \cite{AmmNie-08,AmmNie-09,AmmNie-11}. It is then simple to deduce the usual statement when there is strong convergence. The method in \cite{LewNamRou-13} thus shows in particular that the de Finetti theorem for general spaces can be deduced from the version in finite 
dimensional spaces, which in turn follows from the quantitative bounds we discuss in the sequel. This has the advantage of providing a more constructive (but longer) proof of the result than that given in the original references \cite{Stormer-69,HudMoo-75} and in \cite{CavFucSch-02}. 

\medskip 
 
\textbf{Acknowledgments.} NR thanks Denis Basko, Thierry Champel and Markus Holzmann for a helpful discussion. We also benefited from interesting exchanges with Matthias Christandl, Isaac Kim and Jan Philip Solovej and financial support from the European Research Council (FP7/2007-2013 Grant Agreement MNIQS 258023) and the ANR (Mathostaq project, ANR-13-JS01-0005-01). NR and PTN acknowledge the hospitality of the Institute for Mathematical Science of the National University of Singapore.
 
\section{An explicit construction and the associated estimates}\label{sec:CKMR}

\subsection{Motivation and heuristics}\label{sec:motiv}

When the one-body Hilbert space $\gH$ is finite dimensional, one can base an explicit construction on Schur's lemma: by rotational invariance of the (normalized) Haar (uniform) measure $du$ on the unit sphere $S\gH$ we have
\begin{equation}\label{eq:Schur}
 \dim \gH^N \int_{S\gH} | u^{\otimes N} \rangle  \langle u^{\otimes N} | \, du = \1_{\gH^N} 
\end{equation}
where   
$$\dim \gH^N = {N+d-1 \choose {d-1}}\quad\text{and}\quad \dim\gH=d.$$ 
The decomposition \eqref{eq:Schur} is a coherent state representation \cite{KlaSka-85,ZhaFenGil-90}, and in this respect, one may rephrase the problem we are looking at as the quest for an \emph{upper symbol} associated to a given (mixed) state $\Gamma$ on $\gH^N$. That is we may look for a positive measure $d\muu$, absolutely continuous with respect to $du$, such that the identity
\begin{equation}\label{eq:up symb}
\Gamma  = \int_{S\gH} d\muu (u) | u^{\otimes N} \rangle  \langle u^{\otimes N} |, 
\end{equation}
or at least an approximation of it, holds. Of course there need not be such a \emph{positive} upper symbol\footnote{Strictly speaking the upper symbol designates the density of $d\muu$ with respect to $\dim \gH^N\: du$. We shall abuse language throughout this discussion.} for every state (but there is always one if $d\muu$ is allowed to be a signed measure~\cite{Simon-80}). This point of view is helpful nevertheless because there is always a (positive) \emph{lower symbol} (or Husimi function) associated to the state in the representation \eqref{eq:Schur}:
\begin{equation}\label{eq:low symb}
d\mul (u) :=\dim \gH^N \pscal{u^{\otimes N},\Gamma u^{\otimes N}} du. 
\end{equation}
Now, the limit $N\to \infty$ is reminiscent of a semi-classical limit, in which one may expect that upper and lower symbols will approximately coincide \cite{Lieb-73b,Simon-80}. A good guess is then to define the candidate de Finetti measure as being the lower symbol of $\Gamma$ in the representation~\eqref{eq:Schur}. 

We can even motivate this choice more explicitly: assume that the state in question does have a positive upper symbol to begin with, i.e. that \eqref{eq:up symb} holds with $\muu$ a positive measure. We are then looking for an approximate expression of the upper symbol in terms of the state itself. Let us compute the lower symbol\footnote{The ``lower symbol of the upper symbol''.} of~\eqref{eq:up symb}:
\begin{align*}
d\mul (u) &=\dim \gH^N \pscal{u^{\otimes N},\Gamma u^{\otimes N}} du \\
&= \dim \gH^N \left( \int_{S\gH} |\langle u, v \rangle| ^{2N} d\muu (v) \right) du.
\end{align*}
It is clear that for large $N$, the contribution to the above integral of vectors $v$ that are not exactly colinear to $u$ will be negligible, and therefore that the integral will simply converge to $d\muu (u)$, which means that
$$ d\mul \to d\muu $$
as measures when $N\to \infty$. If the upper symbol was not positive to begin with, this still shows that it can be approximated by a positive measure (the lower symbol) when $N$ is large. This should not be surprising because (as we used above) $\langle u ^{\otimes N}, v  ^{\otimes N}\rangle_{\gHNs} \to 0$ if $u$ and $v$ are not colinear and $N\to \infty$. The coherent state basis $(u ^{\otimes N})_{u\in \S \gH}$ thus tends to be ``less and less overcomplete'' in this limit, i.e. closer and closer to a true basis. Then the upper symbol should become positive simply by positivity of the state. 

All this suggests that the lower symbol should be a good approximation to the upper symbol. Of course the lower symbol also has the advantage of being always well-defined as an explicit, positive, function of the state. It is thus natural to take 
\begin{equation}\label{eq:Husimi deF}
\boxed{d\mu_N (u) :=\dim \gH^N \pscal{u^{\otimes N},\Gamma u^{\otimes N}} du} 
\end{equation}
as the candidate de Finetti measure of a given $N$-body state $\Gamma$, and this is the choice of \cite{ChrKonMitRen-07}. In the following we state an estimate of the form \eqref{eq:quant deF}, which confirms its sensibility. 

\subsection{A quantitative de Finetti theorem and an explicit formula}\label{sec:results}
We are interested in relating the density matrices of an $N$-body state $\Gamma_N$
$$ \gamma_N^{(k)}:=\Tr_{k+1\to N} \Gamma_N$$
to those of the state defined by taking \eqref{eq:Husimi deF} as an upper symbol in the representation \eqref{eq:Schur} 
\begin{equation}\label{eq:def representation}
\widetilde \gamma _N^{(k)} :=\int_{S\gH} d\mu_N(u) |u^{\otimes k} \rangle \langle u^{\otimes k}|,\quad\text{with}\quad d\mu_N (u) :=\dim \gH^N \pscal{u^{\otimes N},\Gamma_N u^{\otimes N}}du. 
\end{equation}
The main estimate is the following.

\begin{theorem}[\textbf{Quantitative quantum de Finetti in finite dimension}]\label{thm:CKMR} \mbox{}\\
Let $\gH$ be a Hilbert space of dimension $d$. For every state $\Gamma_N$ on $\gH^N$, let $\widetilde \gamma _N^{(k)}$ be defined as in \eqref{eq:def representation}. Then for every $k=1,2,...,N$ we have
\begin{equation} \label{eq:error-CKMR-improved}
\boxed{\Tr_{\gH^k} \Big| \gamma_N ^{(k)} - \widetilde \gamma _N^{(k)}  \Big| \le 
\begin{cases}
2 &\text{if $N\leq 2kd$},\\
\displaystyle\frac{2kd}{N-kd}&\text{if $N>2kd$}.\\
\end{cases}}
\end{equation}
\end{theorem}

\begin{remark}\label{rm:Stormer}
As claimed before, such a statement implies the St\o{}rmer-Hudson-Moody theorem recalled in~\eqref{eq:deF strong}, for finite dimensional spaces. Indeed, starting from~\eqref{eq:reduced matrices} one can always extract convergent subsequences from reduced density matrices as in~\eqref{eq:weak conv}, and the convergence is strong since $\gH ^k$ is finite dimensional. The measure $\mu_N$ on the other hand is a probability over a compact set, so along a further subsequence $\mu_N \to \mu$, a probability, weakly as measures. Then, for any $V_k\in \gH ^k$
$$ \langle V_k , \widetilde \gamma _N^{(k)} V_k \rangle = \int_{S\gH} \left| \langle u ^{\otimes k}, V_k \rangle \right| ^2 d\mu_N (u) \to \int_{S\gH} \left| \langle u ^{\otimes k}, V_k \rangle \right| ^2 d\mu (u)$$
since $u\mapsto \left| \langle u ^{\otimes k}, V_k \rangle \right| ^2$ is clearly a continuous function of $u$. In finite dimension this implies strong convergence of the reduced density matrices of $\Gammat_N$:
$$ \widetilde \gamma _N^{(k)} \to \int_{S \gH} |u ^{\otimes k} \rangle \langle u ^{\otimes k} |d\mu(u)$$
and thus we deduce~\eqref{eq:deF strong} from the estimate~\eqref{eq:error-CKMR}. \hfill\qed
\end{remark}

Recall that $\tr \gamma_N ^{(k)} = \tr \widetilde \gamma _N^{(k)} =1$, and therefore the bound
$$\Tr_{\gH^k} \Big| \gamma_N ^{(k)} - \widetilde \gamma _N^{(k)}  \Big| \le 2$$
is an obvious consequence of the triangle inequality for the trace-norm. On the other hand, we have
$$ \frac{\delta}{1-\delta}\leq 2\delta, \quad \forall \: 0\leq \delta\leq \frac12,$$
and this yields the simpler bound
\begin{equation} \label{eq:error-CKMR}
\Tr_{\gH^k} \Big| \gamma_N ^{(k)} - \widetilde \gamma _N^{(k)}  \Big| \le 
\frac{4kd}{N}.
\end{equation}
The slightly weaker inequality~\eqref{eq:error-CKMR} is the one that was proved by Christandl, König, Mitchison and Renner in \cite{ChrKonMitRen-07} (with a different convention for the trace norm, which is divided by~$2$). 
In Section \ref{sec:proof CKMR}, we provide a proof of the estimate~\eqref{eq:error-CKMR-improved}, which is very similar but not identical to that of \eqref{eq:error-CKMR} in~\cite{ChrKonMitRen-07}. While the proof in \cite{ChrKonMitRen-07} uses a ``lifting up" argument (namely going from the $N$ particle space to the $(N+k)$ particle space), we rather use a ``lifting down" argument. Related arguments and estimates may also be found in \cite{Chiribella-11}, where the bound is improved to $2kd/N$.

\begin{remark}\label{rm:infinite dimension}
It is clear from the discussion is Subsection~\ref{sec:motiv} that the construction used to prove Theorem~\ref{thm:CKMR} works only for finite dimensional spaces. It is an open problem to obtain an estimate with a better $d$-dependence, in particular one that would apply to infinite dimensional systems. 

As we mentioned at the end of Section~\ref{sec:intro}, the construction used here is useful, when combined with localization methods, to provide a constructive proof of the infinite dimensional quantum de Finetti theorem of St\o{}rmer-Hudson-Moody mentioned in the Introduction, see~\cite[Section 2]{LewNamRou-13}. This proof, based on the finite dimensional constructions discussed here, has applications to infinite dimensional settings~\cite{LewNamRou-14,LewNamRou-14b}.\hfill\qed
\end{remark}

\medskip
% 
% Our main contribution in this paper is to provide an explicit formula giving the density matrices  $\widetilde \gamma _N^{(k)}$ as a function of the density matrices of the original state:   

We next state the explicit formula relating $\tilde\gamma^{(k)}_N$ to $(\gamma^{(\ell)}_N)_{\ell= 0 \ldots k}$. It was first derived by Chiribella \cite{Chiribella-11} (along with the associated Remark \ref{rm:non opt bound}) and we will provide a different proof.

\begin{theorem}[\textbf{Explicit formula for $\tilde\gamma^{(k)}_N$}] \label{thm:CKMR-identity} \mbox{}\\
For the density matrices $\gamma_N^k$ and $\widetilde \gamma _N^{(k)}$ given above, we have
\begin{equation}\label{eq:CKMR exact}
\widetilde \gamma _N^{(k)} = {{N+k+d-1}\choose k}^{-1}\sum_{\ell=0}^{k} {N \choose \ell}  \gamma_N^{(\ell)} \otimes _s \1_{\gH^{k-\ell}}
\end{equation}
with the convention that
$$ \gamma_N^{(\ell)} \otimes _s \1_{\gH^{k-\ell}}= \frac{1}{\ell!\,(k-\ell)!}\sum_{\sigma\in S_k} (\gamma_N^\ell)_{\sigma(1),...,\sigma(\ell)} \otimes (\1_{\gH^{k-\ell}})_{{\sigma(\ell+1)},...,{\sigma(k)}}.$$
\end{theorem} 

\begin{remark}\label{rm:non opt bound}

A bound of the form \eqref{eq:quant deF} can also be deduced easily from the formula~\eqref{eq:CKMR exact}. Indeed, from the representation in Theorem \ref{thm:CKMR-identity} we may write 
\begin{equation}
\widetilde \gamma _N^{(k)} - \gamma_N ^{(k)}  = ( C(d,k,N) - 1) \gamma_N ^{(k)} + B = -A + B \label{eq:estim CKMR}
\end{equation}
where 
\[
C(d,k,N) = \frac{(N+d-1)!}{(N+k+d-1)!} \frac{N!}{(N-k)!} < 1,  
\]
and $A,B$ are non-negative operators. Since from \eqref{eq:estim CKMR} it is clear that $\Tr(-A+B)=0$, the triangle inequality gives
\[
\Tr \Big|\widetilde \gamma _N^{(k)} - \gamma_N ^{(k)} \Big| \leq \tr A + \tr B = 2 \Tr A = 2 (1- C(d,k,N)).
\]
By the elementary inequality 
\begin{align*}
C(d,k,N) &=  \prod_{j=0} ^{k-1} \frac{N-j}{N+j+d}\ge  \left( 1 - \frac{2k + d -2}{N + d + k - 1}\right) ^k \geq 1 -k \frac{2k + d -2}{N + d + k - 1}
\end{align*}
we find that  
\begin{equation} \label{eq:error-NLR}
\Tr \Big| \gamma_N ^{(k)} - \widetilde \gamma _N^{(k)} \Big| \le \frac{2 k(d+2k)}{N}.
\end{equation}
The $k$-dependence in \eqref{eq:error-NLR} is not as good as that of \eqref{eq:error-CKMR-improved}, but at least for fixed $k$ and $d\ll N$ we recover the same dependence on $d/N$. Recall that fixed $k$ and large $N$ is the relevant limit for studying mean-field approximations of many-body systems. Only $k=2$ is needed for systems comprising two-body interactions.\hfill \qed

\end{remark}

The main idea behind our proof of Theorem \ref{thm:CKMR-identity} is that the density matrices of $\widetilde \gamma _N^{(k)}$ turn out to be defined via an anti-Wick (anti-normal order of creation and annihilation operators) quantization, whereas the original density matrices $\gamma _N^{(k)}$ are of course defined by a Wick (normal order) quantization. Once this has been observed, the proof of \eqref{eq:CKMR exact}, given in Section \ref{sec:Wick}, consists in using the Canonical Commutation Relation repeatedly, with the upshot that, since there are many particles but few available degrees of freedom ($d\ll N$), annihilation and creation almost commute: their commutators are of order $1$ whereas the operators themselves should roughly be of order $\sqrt{N/d}$. The connection between quantum de Finetti theorems for bosonic states and the Wick versus anti-Wick quantization issue was inspired to us by the approach of Ammari and Nier \cite{Ammari-HDR,AmmNie-08,AmmNie-09,AmmNie-11}. We also remark that, 
independently of our work, Lieb and Solovej \cite{LieSol-13} use a formula very similar to \eqref{eq:CKMR exact} in their investigation of the classical entropy of quantum states. 

\medskip

\section{Proof of the main estimate, Theorem \ref{thm:CKMR}}\label{sec:proof CKMR}

In this section we give the proof of the bound~\eqref{eq:error-CKMR-improved}, following ideas from~\cite{ChrKonMitRen-07}. For simplicity of writing, we only deal with pure states $\Gamma_N=|\Psi_N\rangle\langle\Psi_N|$, as it is clear that the general case follows from the triangle inequality. 

Denote $P_u:=|u \rangle \langle u|$ for every $u\in S\gH$. Note that $P_u^{\otimes k}=|u^{\otimes k} \rangle \langle u^{\otimes k}|$ for every $k\in \mathbb N$. Thus for every bounded  operator $A$ on $\gH^k$, using Schur's formula (\ref{eq:Schur}) we find that
$$ \Tr [A \gamma_N^{(k)}]= \langle \Psi_N, (A \otimes \1_{\gH ^{N-k}}) \Psi_N\rangle = \dim \gH^N \int_{S\gH}\langle \Psi_N, P_u ^{\otimes N} (A \otimes \1_{\gH ^{N-k}}) \Psi_N\rangle du .$$   
On the other hand, by the definition of $\widetilde\gamma_N^{(k)}$,  
\begin{align*} \Tr [A \widetilde\gamma_N^{(k)}] &=  \dim \gH^N \int_{S\gH} \langle Au^{\otimes k}, u^{\otimes k} \rangle . |\langle \Psi_N, u ^{\otimes N} \rangle|^2 du \\
&=\dim \gH^N \int_{S\gH} \big\langle (A\otimes \1_{\gH ^{N-k}}) u^{\otimes N}, u^{\otimes N} \big\rangle . \big|\langle \Psi_N, u ^{\otimes N} \rangle\big|^2 du \\
&=\dim \gH^N \int_{S\gH} \big\langle \Psi_N, P_u^{\otimes N} (A\otimes \1_{\gH ^{N-k}}) P_u^{\otimes N} \Psi_N \big\rangle du .
\end{align*}   
Thus
\begin{align} \label{eq:diff-gammaN}
\Tr [A (\gamma_N^{(k)}-\widetilde\gamma_N^{(k)})] =
\dim \gH^N \int_{S\gH} \big\langle \Psi_N, P_u^{\otimes N} (A \otimes \1_{\gH ^{N-k}}) (\1_{\gH ^{N}}-P_u^{\otimes N}) \Psi_N \big\rangle du .
\end{align}   
Using 
\begin{equation*} 
P_u^{\otimes N} (A \otimes \1_{\gH ^{N-k}}) (\1_{\gH ^{N}}-P_u^{\otimes N}) = P_u^{\otimes k} A(\1_{\gH ^{k}}-P_u^{\otimes k})  \otimes P_u^{\otimes N-k}  
%\\&=  \big[A\otimes \1-(\1-P_u^{\otimes k}) A \big] (\1-P_u^{\otimes k})\otimes P_u^{\otimes N-k},
\end{equation*}
%the triangle inequality $\| A\otimes \1-(\1-P_u^{\otimes k}) A \| \le 2 \|A\|$ and Schur's formula again, 
we find that
\begin{multline*}
 \int_{S \gH } P_u^{\otimes N} (A \otimes \1_{\gH ^{N-k}}) (\1_{\gH ^{N}}-P_u^{\otimes N}) du 
 \\ = A\otimes \1_{\gH ^{N-k}} \int_{S\gH} \left( \1_{\gH ^{k}}- P_u ^{\otimes k}\right) \otimes P_u ^{\otimes N-k}du 
  - \int_{S\gH} \left( \1_{\gH ^{k}}- P_u ^{\otimes k}\right) A \left( \1_{\gH ^{k}}- P_u ^{\otimes k}\right) \otimes P_u ^{\otimes N-k} du.
\end{multline*}
Then, using Schur's formula \eqref{eq:Schur} in $\gH ^N$ and $\gH ^{N-k}$ we have
\begin{equation}\label{eq:rewrite dif}
\int_{S\gH} \left( \1_{\gH ^{k}}- P_u ^{\otimes k}\right) \otimes P_u ^{\otimes N-k} du = \left( \left(\dim \gH ^{N-k} \right) ^{-1} - \left(\dim \gH ^N \right) ^{-1}\right)\1_{\gH ^{N}}
\end{equation}
and since 
\[
 \left( \1_{\gH ^{k}}- P_u ^{\otimes k}\right) A \left( \1_{\gH ^{k}}- P_u ^{\otimes k}\right) \leq \Vert A \Vert \left( \1_{\gH ^{k}}- P_u ^{\otimes k}\right)
\]
we conclude
% \begin{align*} \Big\|\int_{S\gH} P_u^{\otimes N} A \otimes 1 (1-P_u^{\otimes N}) du \Big\| &\le 2\|A\| \int_{S\gH} (\1-P_u^{\otimes k})\otimes P_u^{\otimes N-k} du \\
% &= 2\|A\| \Big[(\dim \gH^{N-k})^{-1} - (\dim \gH^{N})^{-1} \Big]. 
% \end{align*}   
from \eqref{eq:diff-gammaN} and \eqref{eq:rewrite dif} that
\begin{align*}
\Big| \Tr [A (\gamma_N^{(k)}-\widetilde\gamma_N^{(k)})] \Big| \le 2 \|A\| \Big(  \frac{\dim \gH^N}{\dim \gH^{N-k}} -1\Big)
\end{align*}
for every bounded operator $A$ on $\gH^k$. This implies the upper bound   
\begin{align} \label{eq:trace-bound-cN}
\Tr \Big|\gamma_N^{(k)}-\widetilde\gamma_N^{(k)} \Big| \le 2 \Big( \frac{\dim \gH^N}{\dim \gH^{N-k}} -1\Big).
\end{align}
Finally, due to Bernoulli's inequality we have 
\begin{align*}
\frac{\dim \gH^{N-k}}{\dim \gH^{N}} &= \frac{{{N+d-k-1} \choose {d-1}}}{{ {N+d-1} \choose {d-1}}}= \frac{(N-k+1)...(N-k+d-1)}{(N+1)...(N+d-1)} \\
& = \Big( 1- \frac{k}{N+1} \Big) ... \Big( 1- \frac{k}{N+d-1}\Big) \ge \Big(1- \frac{k}{N}\Big)^d \ge 1- \frac{dk}{N},
\end{align*}
which implies that, in the case $dk < N$
\begin{align*}
\frac{\dim \gH^N}{\dim \gH^{N-k}} -1 \le \Big( 1- \frac{dk}{N} \Big)^{-1}-1 = \frac{dk}{N-dk}.
\end{align*}
The desired estimate (\ref{eq:error-CKMR-improved}) then follows immediately from (\ref{eq:trace-bound-cN}).\qed

\section{Proof of the explicit formula, Theorem~\ref{thm:CKMR-identity}}\label{sec:Wick}

Our proof of Theorem \ref{thm:CKMR-identity} is based on the fact that $\widetilde \gamma _N^{(k)}$ turns out to be linked to an anti-Wick representation, while $\gamma _N^{(k)}$ is defined via a standard Wick representation. The difference between $\gamma _N^{(k)}$ and $ \widetilde \gamma _N^{(k)}$ can then be computed by comparing Wick and anti-Wick representation, that is, by looking at the difference between normal ordered and anti-normal ordered polynomials in annihilation and creation operators. 

Recall that for every $f_{k}\in \gH$, we can define the creation operator $a^*(f_{k}): \gH^{k-1} \to \gH^{k}$ by
$$
{a^*}({f_{k}})\left( {\sum\limits_{\sigma  \in {S_{k-1}}} {{f_{\sigma (1)}} } \otimes ... \otimes {f_{\sigma (k-1)}}} \right) = (k) ^{-1/2} \sum\limits_{\sigma  \in {S_{k}}} {{f_{\sigma (1)}} }  \otimes ... \otimes {f_{\sigma (k)}}
$$
The annihilation operator $a(f): \gH^{k+1} \to \gH^{k}$ is the adjoint of $a^*(f)$, given by
$$ a(f) \left( {\sum\limits_{\sigma  \in {S_{k+1}}} {{f_{\sigma (1)}} } \otimes ... \otimes {f_{\sigma (k+1)}}} \right) = (k+1) ^{1/2} \sum\limits_{\sigma  \in {S_{k+1}}} \left\langle f,f_{\sigma(1)} \right\rangle {{f_{\sigma (2)}} }  \otimes ... \otimes {f_{\sigma (k)}}$$
for all $f,f_1,...,f_{k}$ in $\gH$. These operators satisfy the {\it canonical commutation relations}
\begin{equation}\label{eq:CCR}
[a(f),a(g)]=0,\quad[a^*(f),a^*(g)]=0,\quad [a(f),a^*(g)]= \langle f,g \rangle_{\gH}. 
\end{equation}

\medskip

We shall need two lemmas. The first one says that any bosonic $k$-body density matrix can be completely determined by its expectation against Hartree states $u^{\otimes k}$. 

\begin{lemma}[\textbf{Expectations in Hartree vectors determine the state}]\label{le:uk-g-uk=0}\mbox{} \\
% Let $\gH$ be an arbitrary separable Hilbert space (finite or infinite dimensional). 
If a trace class self-adjoint operator $\gamma^{(k)}$ on $\gH^k$ satisfies  
\bq \label{eq:uk-g-uk=0}
 \langle u^{\otimes k}, \gamma^{(k)} u^{\otimes k} \rangle =0\qquad \text{for ~all}~u\in \gH,
 \eq
 then $\gamma^{(k)} \equiv 0$.
\end{lemma}

In connection with the discussion in Section \ref{sec:motiv}, this result says that the state is uniquely determined by its lower symbol. This is a well-known fact even in  more abstract settings \cite{Klauder-64,Simon-80}. For the reader's convenience, a standard proof is given in Appendix \ref{app:Hartree-determine}. 
%Note that the inequality $\langle u^{\otimes k}, \gamma^{(k)} u^{\otimes k} \rangle \ge 0$ for all $u\in \gH$ does {\em not} imply that $\gamma^{(k)}\ge 0$, except when $k=1$.

\medskip

In the second lemma, we compare normal and anti-normal ordering of creation and annihilation operators. Thanks to Lemma \ref{le:uk-g-uk=0} we need only do this for a single mode~$v\in \gH$.

\begin{lemma}[\textbf{Wick versus anti-Wick representations}]\label{le:Wick A Wick}\mbox{}\\
Let $v\in S\gH$ 
%be a normalized vector in an arbitrary separable Hilbert space, 
with associated creation and annihilation operators $a ^* (v)$ and $a(v)$. Then 
\begin{equation}\label{eq:Wick A Wick}
a(v) ^n a ^*(v) ^n = \sum_{k=0} ^n \binom{n}{k} \frac{n!}{k!} a ^*(v) ^k a (v) ^k \mbox{ for any } n\in \NN. 
\end{equation}
\end{lemma}
% Note that the coefficient of $a ^*(v) ^n a (v) ^n$ in \eqref{eq:Wick A Wick} is just $1$. 

\begin{proof}%[Proof of Lemma \ref{le:Wick A Wick}] 
Recall that the $n$-th Laguerre polynomial is given by   
\[
L_n (x) = \sum_{k=0} ^n \binom{n}{k} \frac{(-1) ^k}{k!} x ^k
\]
and these polynomials satisfy the relation
\[
(n+1) L_{n+1} (x) = (2n +1) L_n (x) - x L_n (x) - n L_{n-1} (x).  
\]
The identity \eqref{eq:Wick A Wick} is equivalent to 
\begin{equation}\label{eq:Wick A Wick proof}
a(v) ^n a ^*(v) ^n = \sum_{k=0} ^n c_{n,k} \, a^*(v) ^k a (v) ^k 
\end{equation}
where the $c_{n,k}$'s are the coefficients of the polynomial
$$\tilde{L}_n (x) := n! \, L_n (-x).$$

We prove \eqref{eq:Wick A Wick proof} by induction on $n$. Note first that the CCR \eqref{eq:CCR} immediately gives \eqref{eq:Wick A Wick} for $n=1$, while it is easy to see that 
\begin{equation}\label{eq:Wick A Wick 2}
a(v) ^2 a ^*(v) ^2 = a ^*(v) ^2 a (v) ^2 + 4 a ^* (v) a (v) + 2 
\end{equation}
by a repeated use of the CCR. This is \eqref{eq:Wick A Wick} for $n=2$, so we simply need an induction formula giving $a(v) ^{n+1} a ^*(v) ^{n+1}$ as a function of $a(v) ^n a ^*(v) ^n$ and $a(v) ^{n-1} a ^*(v) ^{n-1}$. We claim that 
\begin{equation}\label{eq:induction}
a(v) ^{n+1} a ^*(v) ^{n+1} = a^* (v)  a(v) ^n a ^*(v) ^n  a(v) + (2n+1) a(v) ^n a ^*(v) ^n - n ^2  a(v) ^{n-1} a ^*(v) ^{n-1}.
\end{equation}
Note the order of creation and annihilation operators in the first term: knowing a normal ordered representation of $a(v) ^n a ^*(v) ^n$ we can deduce a normal ordered representation for this term. Since the modified Laguerre polynomials satisfy
\[
\tilde{L}_{n+1} (x) = (2n +1) \tilde{L}_n (x) + x \tilde{L}_n (x) - n ^2 \tilde{L}_{n-1} (x) 
\]
it is then clear that \eqref{eq:Wick A Wick proof} follows from \eqref{eq:induction}.

We now prove our claim \eqref{eq:Wick A Wick proof}. We use the relations 
\begin{align}\label{eq:CCR n}
a(v) a ^* (v) ^n &= a ^* (v) ^n a(v) + n  a ^* (v) ^{n-1} \nonumber \\
a(v) ^n a ^* (v) &= a ^* (v)  a(v) ^n + n  a (v) ^{n-1}
\end{align}
that follow from the CCR. Then 
\begin{align*}
a^* (v) a(v) ^n a ^*(v) ^n  a(v) &= a(v) ^n a ^* (v) ^{n+1} a (v) - n a (v) ^{n-1} a ^* (v) ^n a(v)\\
& = a(v) ^{n+1} a ^*(v) ^{n+1} - (n + 1) a(v) ^n a ^*(v) ^n \\
&- n a(v) ^n a ^*(v) ^n +  n ^2  a(v) ^{n-1} a ^*(v) ^{n-1},
\end{align*}
which is \eqref{eq:Wick A Wick proof}.
\end{proof}

% Although the lemma can be proved using an elementary induction and probably well-known, in Appendix \ref{app:Wick-antiWick} we provide a short proof using Laguerre polynomials for completeness.

Now we are able to give the
\begin{proof}[Proof of Theorem \ref{thm:CKMR-identity}]
Clearly it is sufficient to consider only the case of a pure state $|\Psi_N \rangle \langle \Psi_N |$. Using Lemma \ref{le:uk-g-uk=0}, the $k$-particle density matrix of $\Psi_N$ is uniquely defined by
\bqq 
 \langle v^{\otimes k}, \gamma_N^{(k)} v^{\otimes k} \rangle =\frac{(N-k)!}{N!} \langle \Psi_N, a^*(v)^k a(v)^k \Psi_N \rangle
 \eqq
for all $v\in \gH$ such that $\norm{v}=1$. In contrast, the $ \widetilde \gamma _N^{(k)}$ satisfies a similar formula but with the order of the creation and annihilation operators reversed:
\begin{align*}
\langle v^{\otimes k}, \widetilde \gamma _N^{(k)} v^{\otimes k} \rangle &= \dim \gH^N \int_{S\gH} du |\langle u^{\otimes N}, \Psi_N \rangle|^2 | \langle  u^{\otimes k},v^{\otimes k} \rangle|^2 \\
&= \dim \gH^N \int_{S\gH} du |\langle u^{\otimes (N +k)}, v^{\otimes k} \otimes \Psi_N \rangle|^2 \\
&= \frac{N!}{(N+k)!} \dim \gH^N \int_{S\gH} du |\langle u^{\otimes (N +k)}, a^*(v)^{k} \Psi_N \rangle|^2 \\
&=  \frac{N!}{(N+k)!}\frac{\dim \gH^N}{\dim \gH^{N+k}} \| a^*(v)^{k} \Psi_N \|^2 \\
&= \frac{(N+d-1)!}{(N+k+d-1)!} \langle \Psi_N, a(v)^k a^*(v)^{k} \Psi_N \rangle
\end{align*}
where we used Schur's formula \eqref{eq:Schur} for the fourth equality. There only remains to use Lemma \ref{le:Wick A Wick}:
\begin{align*}
\frac{(N+k+d-1)!}{(N+d-1)!} \langle v^{\otimes k}, \widetilde \gamma _N^{(k)} v^{\otimes k} \rangle &= \langle \Psi_N, a(v)^k a^*(v)^{k} \Psi_N \rangle \\
&= \sum_{\ell=0}^{k}  \binom{k}{\ell} \frac{k!}{\ell!} \langle \Psi_N,  a^*(v)^\ell a(v)^\ell \Psi_N \rangle\\
& = \sum_{\ell=0}^{k} {N \choose \ell} \binom{k}{\ell} k! \langle v^{\otimes \ell},   \gamma_N^{(\ell)} v^{\otimes \ell}\rangle
 \end{align*}
and \eqref{eq:CKMR exact} then follows from Lemma \ref{le:uk-g-uk=0}.
\end{proof}

\appendix

%%%%%%%%%%%%%%%%%%%%%%%%%%%%%%%%%%%%%%%%%%%%
%%%%%%%%%%%%%%%%%%%%%%%%%%%%%%%%%%%%%%%%%%%%
\section{Expectations in Hartree vectors determine the state} \label{app:Hartree-determine}
In the following we use the symmetric tensor product 
$$\Psi_k\otimes_s\Psi_\ell(x_1,...,x_{k})=\frac{1}{\sqrt{\ell!(k-\ell)!k!}}\sum_{\sigma\in {S}_{k}}\Psi_\ell(x_{\sigma(1)},...,x_{\sigma(\ell)})\Psi_{k-\ell}(x_{\sigma(\ell+1)},...,x_{\sigma(k)})
$$
of two functions $\Psi_\ell \in \gH^{\ell}$ and $\Psi_{k-\ell}\in\gH^{k-\ell}$. Note that for every $f\in \gH$,
$$
f\otimes_s \Psi_\ell = a^*(f) \Psi_\ell.
$$
\begin{proof}[Proof of Lemma \ref{le:uk-g-uk=0}] By replacing $u$ by $u+tv$ in (\ref{eq:uk-g-uk=0}) and taking the derivative with respect to $t$, we obtain
 \bqq
 \langle v \otimes_s u^{\otimes (k-1)}, \gamma^{(k)} v \otimes_s u^{\otimes (k-1)} \rangle =0
 \eqq
 for all $u,v\in \gH$. Taking $v$ in the form $v=v_1\pm \widetilde v_1$ and then $v = v_1 \pm i \widetilde v_1$ we deduce
  \bqq
\langle v_1 \otimes_s u^{\otimes (k-1)}, \gamma^{(k)}  \widetilde v_1 \otimes_s u^{\otimes (k-1)} \rangle =0
 \eqq
 for all $u,v_1,\widetilde v_1\in \gH$. We may then again replace $u$ by $u+tv$ in the above, and take the derivative with respect to $t$. Repeating this process $k$ times we conclude that 
   \bqq
\langle v_1 \otimes_s v_2 \otimes_s \ldots \otimes_s v_k, \gamma^{(k)} \widetilde v_1 \otimes_s \widetilde v_2 \otimes_s \ldots \otimes_s \widetilde v_k \rangle =0
 \eqq
for all $v_j, \widetilde v_j \in \gH$. Since vectors of the forms  $v_1 \otimes_s v_2 \otimes_s \ldots \otimes_s v_k$
form a complete basis for $\gH^k$, we conclude that $\gamma^{(k)}\equiv 0$.
\end{proof}

% \section{Schur's formula} \label{app:Schur}
% 
% \begin{proof}[Proof of Schur's formula (\ref{eq:Schur}).] Fix a vector $u_0\in S\gH$. For every $v\in \gH$, if we choose a unitary mapping $U_v$ such $U_v v =u_0$, then using the unitary invariance of the Haar measure $du$ we have
% $$ \int_{S\gH} du | \langle u,v \rangle |^{2N}=\int_{S\gH} du | \langle U_v u,U_v v \rangle |^{2N}=\int_{S\gH} du | \langle u,u_0 \rangle |^{2N}=: \Lambda_0.$$
% where $\Lambda_0$ is a constant dependent only on $u_0$. By Lemma \ref{le:uk-g-uk=0}, we conclude that
% \bq \label{eq:Schur-pre}
% \int_{S\gH} du | u^{\otimes N} \rangle  \langle u^{\otimes N} |=\Lambda_0 \1_{\gH^N}.
% \eq
% From this identity, by taking the trace over $\gH^N$ we find that $\Lambda_0 = (\dim \gH^N )^{-1}$, which concludes the proof. 
% \end{proof}

%\section{Wick versus anti-Wick representations}\label{app:Wick-antiWick}

%Now we give a proof of Lemma \ref{le:Wick A Wick}. 

%%%%%%%%%%%%%%%%%%%%%%%%%%%%%%%
%%%%%%%%%%%%%%%%%%%%%%%%%%%%%%
% \bibliographystyle{siam}
% \bibliography{biblioNR_Peccot_08-14}

\end{document}